\documentclass{ws-ijgmmp}

\usepackage{amssymb}
\usepackage{amsmath}
\usepackage{amsfonts}
\usepackage{bbm}
\usepackage{mathrsfs}
\usepackage[all,cmtip]{xy}

\def\nn{\nonumber}

\def\bbR{\mathbb{R}}
\def\bbC{\mathbb{C}}
\def\bbN{\mathbb{N}}

\def\EE{\mathcal{E}}

\def\Hom{\mathrm{Hom}}
\def\End{\mathrm{End}}
\def\Con{\mathrm{Con}}

\def\Tr{\mathrm{Tr}}
\def\Der{\mathrm{Der}}

\def\id{\mathrm{id}}

\def\dd{\mathrm{d}}

\def\sc{\mathrm{sc}}
\def\dim{\mathrm{dim}}

\def\1{\mathbbm{1}}
\def\g{\mathfrak{g}}

\def\nab{\nabla}

\newcommand{\omi}[1]{\buildrel { \buildrel{#1}\over{\vee} } \over .}

\begin{document}

\markboth{Alexander Schenkel}
{Module parallel transports in fuzzy gauge theory}

%
\catchline{}{}{}{}{}
%

\title{
MODULE PARALLEL TRANSPORTS IN FUZZY GAUGE THEORY
}

\author{
ALEXANDER SCHENKEL
}

\address{Fachgruppe Mathematik, Bergische~Universit\"at~Wuppertal,\\
Gau\ss stra\ss e 20, 42119 Wuppertal, Germany.\\
\email{schenkel@math.uni-wuppertal.de} }

\maketitle

\begin{history}
\received{(Day Month Year)}
\revised{(Day Month Year)}
\end{history}

\begin{abstract}
In this article we define and investigate a notion of parallel transport on finite projective modules over finite matrix algebras. Given a derivation-based differential calculus on the algebra and a connection on the module, we construct for every derivation $X$ a module parallel transport, which is a lift to the module of the one-parameter group of algebra automorphisms generated by $X$. This parallel transport morphism is determined uniquely by an ordinary differential equation depending on the covariant derivative  along $X$. Based on these parallel transport morphisms, we define a basic set of gauge invariant observables, i.e.\ functions from the space of connections to the complex numbers. For modules equipped with a hermitian structure, we prove that this set of observables is separating on the space of gauge equivalence classes of hermitian connections. This solves the gauge copy problem for fuzzy gauge theories.
\end{abstract}



\section{\label{sec:intro}Introduction}
In noncommutative geometry, gauge theories are typically described in terms of connections on
finite projective modules $\EE$ over the noncommutative algebra $A$, see e.g.~Landi's book \cite{Landi:1997sh} for an
introduction. This description requires the notion of a differential calculus on $A$, generalizing differential forms
and the exterior differential to the noncommutative setting. Depending on the algebra $A$, there are different choices for
a calculus at hand, for example the twist deformed differential calculus \cite{Aschieri:2005zs}
 in deformation quantization using Drinfeld twists, the bicovariant differential
calculus for quantum groups \cite{Woronowicz:1989rt} and the derivation-based differential calculus
\cite{DuVi:88,DuViMich:94,DuViMich:96,DuViLecture1,DuViLecture2,Masson:2008uq}, 
being particularly useful for finite matrix algebras
 \cite{DuboisViolette:1988ir}, such as the fuzzy sphere \cite{Madore:1991bw}.

Connections on finite projective modules, and the noncommutative gauge theories they describe,
have been part of many investigations in the past years, leading to interesting results, such as noncommutative instantons
 in deformation quantization \cite{Landi:2006nb} and noncommutative 
 monopoles on fuzzy spaces \cite{Grosse:2001qt}, just to name a few.
 These works aim at understanding topological invariants constructed from the curvature of
the connection, but they do not attempt to construct complete sets of observables, i.e.~gauge invariant
 functions from the space of connections to the complex numbers, 
 which are separating on the space of connections modulo gauge transformations.
 Such complete sets of gauge invariant observables are essential for path-integral or operator algebraic approaches 
 to the quantization of gauge theories.

 In commutative gauge theory the curvature of a connection is not sufficient to extract the complete gauge invariant information
 of the connection. This is the well-known gauge copy problem. 
 In nonabelian Yang-Mills theories this is already the case if we choose spacetime to be
  $\bbR^d$  \cite{WuYang}. For certain gauge groups, the gauge copy problem can be 
  solved by considering the set of all Wilson loops as
  the basic observables  of the gauge theory, see e.g.~\cite{Driver,Sengupta}.
  
 There are examples showing that there is also a gauge copy problem in noncommutative geometry, 
 see e.g.~Proposition \ref{propo:gaugecopyproblem} in this paper, where we review some results of \cite{DuboisViolette:1988ir,DuViLecture1,DuViLecture2,Masson:2008uq}.
 This means that considering only observables which are constructed from the curvature, we can in general not extract the complete
 gauge invariant information of the connection. This calls for a suitable generalization of parallel transport and Wilson loops
 to noncommutative geometry. These studies have been initiated for Moyal-Weyl deformed spaces in earlier works, 
 see e.g.~\cite{Ishibashi:1999hs,Ambjorn:2000cs}.
 
 In our article we will give a purely algebraic definition of parallel transport morphisms on finite projective modules $\EE$ over
 noncommutative algebras $A$. For simplifying the explicit analysis of these morphisms, we study in the present paper
 only the case where $A=M_n(\bbC)$ is a finite matrix algebra, equipped with a derivation-based differential calculus,
 and $\EE$ is any finite projective module over $A$. 
 Given a connection $\nab$ on $\EE$, we can explicitly construct
 for every derivation $X$ a module parallel transport morphism $\Phi^{\nab_X}$, satisfying an ordinary differential equation 
 determined by the covariant derivative $\nab_X$ along $X$. Using the set of all module parallel transport morphisms we define
 a suitable set of gauge invariant observables, i.e.~functions from the space of connections to the complex numbers.
For hermitian finite projective modules we show that this set of observables is separating 
on the space of gauge equivalence classes of hermitian connections.
 
 The structure of this article is as follows: In Section \ref{sec:difcalc}
 we review briefly the basics of derivation-based differential calculi.
 In Section \ref{sec:fuzzy} we study module parallel transports for general finite projective modules
 over $A=M_n(\bbC)$.
 To motivate our definition of module parallel transport we provide  in the \ref{app:commutative} an explicit construction 
 of these morphisms in commutative differential geometry by using the usual concept of parallel transport along curves.


\section{\label{sec:difcalc}Preliminaries on derivation-based differential calculi}
Let $A$ be a unital and associative algebra over $\bbC$, which is not necessarily commutative.
 In \cite{DuVi:88} a general and purely algebraic framework to associate to $A$ a differential calculus
  was introduced. This approach was named the derivation-based differential calculus and has
been further developed in \cite{DuViMich:94,DuViMich:96}, see also the reviews \cite{DuViLecture1,DuViLecture2,Masson:2008uq}.

The basic idea is to consider the space of derivations $\Der(A)$ of the algebra $A$, that is
the $\bbC$-vector space of all $\bbC$-linear maps $X:A\to A$ satisfying, for all $a,b\in A$,
$X(a\,b) = X(a)\,b + a\,X(b)$ and build up a notion of differential geometry on $A$ by 
carefully generalizing the algebraic approach to differential geometry of
commutative smooth manifolds, see e.g.~Koszul's book \cite{Koszul}.
We can equip the vector space $\Der(A)$ with the structure of a Lie algebra by defining, for all $X,Y\in\Der(A)$, the
Lie bracket by the commutator $[X,Y] := X\circ Y - Y\circ X$, where $\circ$ 
denotes the usual composition of endomorphisms $\End_\bbC(A)$.
Furthermore, $\Der(A)$ is a module over the center $\mathcal{Z}(A)$ of the algebra $A$ by defining,
for all $a\in A$, $f\in \mathcal{Z}(A)$ and $X\in\Der(A)$, $(f\cdot X)(a) := f\,X(a)$. 
\begin{definition}\label{defi:diffcalc}
Let $\g\subseteq \Der(A)$ be a sub Lie algebra and a sub $\mathcal{Z}(A)$-module. Let further $\underline{\Omega}_\g^n(A)$, for all $n\in\bbN$, be the vector space of $\mathcal{Z}(A)$-multilinear antisymmetric
maps from $\g^n$ to $A$, $\underline{\Omega}_\g^0(A) := A$ and
\begin{flalign}
\nn \underline{\Omega}^\bullet_\g(A) := \bigoplus_{n\geq 0}\underline{\Omega}_\g^n(A)~.
\end{flalign}
The vector space $\underline{\Omega}^\bullet_\g(A)$ can be equipped with the structure of a differential graded
algebra by defining the product, for all $\omega\in\underline{\Omega}^p_\g(A)$, 
$\eta\in\underline{\Omega}^q_\g(A)$ and $X_1,\dots,X_{p+q}\in\g$,
\begin{multline}
\nn (\omega\,\eta)\big(X_1,\dots,X_{p+q}\big) :=\frac{1}{p!\,q!}\sum\limits_{\sigma \in \mathfrak{S}_{p+q}} (-1)^{\mathrm{sign}(\sigma)}\,
\omega\big(X_{\sigma(1)},\dots, X_{\sigma(p)}\big)\,\\
\times\,\eta\big(X_{\sigma(p+1)},\dots,X_{\sigma(p+q)}\big) 
\end{multline}
and the differential (of degree $1$), for all $\omega\in\underline{\Omega}^p_\g(A)$ and $X_1,\dots,X_{p+1}\in\g$,
\begin{multline}
 \nn \dd \omega\big(X_1,\dots,X_{p+1}\big) :=\sum\limits_{i=1}^{p+1} (-1)^{i+1}\,X_i\Big(\omega\big(X_1,\dots\omi{i} \dots,X_{p+1} \big)\Big) \\
 + \sum\limits_{1\leq i<j\leq p+1} (-1)^{i+j}\,\omega\big([X_i,X_j],\dots \omi{i}\dots \omi{j}\dots,X_{p+1}  \big)~.
\end{multline}
We call $\big(\underline{\Omega}^\bullet_\g(A),\dd\big)$
the {\bf $\g$-restricted derivation-based differential calculus} on $A$.
\end{definition}

Notice that the differential calculus $\big(\underline{\Omega}^\bullet_\g(A),\dd\big)$ depends on 
the choice of $\g$. The prime example indicating the importance of allowing for a proper
 sub Lie algebra and sub $\mathcal{Z}(A)$-module $\g\subset\Der(A)$ is the fuzzy sphere \cite{Madore:1991bw},
 where $\g = \mathfrak{su}_2^\bbC\subset \Der(A)$ is the complexification of $\mathfrak{su}_2$.
 
 Given a $\ast$-algebra $A$ with involution $\ast:A\to A$, we induce an involution on
$\Der(A)$ by defining, for all $X\in\Der(A)$ and $a\in A$, $X^\ast(a) := X(a^\ast)^\ast$.
We assume that $\g\subseteq \Der(A)$ is such that $\ast$ restricts to an involution on $\g$.
We induce further a graded involution on $\underline{\Omega}_\g^\bullet(A)$ by defining, for all 
$\omega\in \underline{\Omega}_\g^p(A)$ and $X_1,\dots, X_p\in \g$,
$\omega^\ast\big(X_1,\dots,X_p\big) := \omega\big(X_1^\ast,\dots,X_p^\ast\big)^\ast $.
As a consequence, we have, for all $\omega\in \underline{\Omega}_\g^p(A) $ and $\eta\in \underline{\Omega}_\g^q(A)$,
$(\omega\,\eta)^\ast= (-1)^{p\,q}~\eta^\ast\omega^\ast$ and  $(\dd\omega)^\ast\ = \dd\omega^\ast$.

We now focus on the case relevant for this paper, namely the algebra
$A=M_n(\bbC)$ of complex-valued $n\times n$-matrices with $n\in\bbN$.
Hermitian conjugation structures $A$ as a $\ast$-algebra. This model has been studied in detail in
\cite{DuboisViolette:1988ir,DuViLecture1,DuViLecture2,Masson:2008uq} and we only summarize the main properties.
The center of $A$ is isomorphic to $\bbC$, via $\bbC\to\mathcal{Z}(A)\,,~\lambda \mapsto \lambda\, \1$, with
$\1\in A$ denoting the unit. Furthermore, all derivations $X\in \Der(A)$ are inner, meaning that
the Lie algebra of traceless complex-valued matrices $\mathfrak{sl}_n\subseteq A$ is isomorphic (as a Lie algebra) to 
$\Der(A)$ via the map $\mathrm{ad}:\mathfrak{sl}_n \to \Der(A)\,,~\gamma \mapsto \mathrm{ad}_{\gamma}$,
where $\mathrm{ad}_\gamma\in\Der(A)$ is the derivation defined by, for all $a\in A$, $\mathrm{ad}_\gamma(a) := [\gamma,a]$.
We denote the inverse of the map $\mathrm{ad}$  by $\theta: \Der(A)\to \mathfrak{sl}_n\subset A$.
With this identification, the Lie algebra (over $\bbR$) of real derivations $\Der_\bbR(A)$ is isomorphic to 
the real sub Lie algebra $\mathfrak{su}_n\subset \mathfrak{sl}_n$ of antihermitian and traceless matrices.
For the $\g$-restricted derivation-based differential calculus we have $
\underline{\Omega}^\bullet_\g(A)\simeq A \otimes_\bbC \bigwedge^\bullet \g^\ast$, where
$\g^\ast$ is the $\bbC$-vector space dual of $\g$. There is a canonical one-form given by
restricting the map $\theta:\Der(A)\to \mathfrak{sl}_n\subset A$ to $\g\subseteq \Der(A)$.
For notational simplicity we use the same symbol $\theta$ for the canonical one-form 
$\underline{\Omega}_\g^1(A) \ni\theta : \g \to A$. The one-form $\theta\in \underline{\Omega}_\g^1(A)$ 
satisfies the Maurer-Cartan equation $\dd \theta - \theta^2 =0$ and its adjoint action
on $A$ is the differential, for all $a\in A$, $\dd a = [\theta,a]$.


\section{\label{sec:fuzzy}Module parallel transport and observables for fuzzy gauge theories}
We study general finite projective modules
over $A=M_n(\bbC)$. Remember that any finite projective right $A$-module $\EE$ is isomorphic to
a right $A$-module of the form $pA^N$, with $N\in\bbN$ being the dimension of an underlying free module
and $p\in M_N(A)$ being a projector, i.e.~$p^2=p$. 
For our studies the following standard lemma turns out to be useful.
\begin{lemma}\label{lem:fpm}
Let $A=M_n(\bbC)$ and let $\EE$ be any finite projective right $A$-module.
Then there exists a $m\in \bbN$, such that $\EE \simeq M_{m,n}(\bbC)$, with $M_{m,n}(\bbC)$ being the
space of complex-valued $m\times n$-matrices. The right $A$-action on the module
 $M_{m,n}(\bbC)$ is given by matrix multiplication from the right, $M_{m,n}(\bbC)\times M_{n}(\bbC)\to M_{m,n}(\bbC)\,,~(s,a)\mapsto
 s\,a$.
\end{lemma}
\begin{proof}
Let $A=M_n(\bbC)$ and let $\EE$ be any finite projective right $A$-module, i.e.~$\EE\simeq p A^N$ for some $N\in\mathbb{N}$. 
We have the following chain of right $A$-module isomorphisms
\begin{flalign*}
A^N &\simeq A\otimes_\bbC \bbC^N \simeq \bbC^n\otimes_\bbC \bbC^n\otimes_\bbC \bbC^N \\
&\simeq
( \bbC^n\otimes_\bbC \bbC^N)\otimes_\bbC \bbC^n \simeq \bbC^{nN}\otimes_\bbC \bbC^n\simeq M_{nN,n}(\bbC)~.
\end{flalign*}
The projector $p\in M_N(A)$ can be regarded as a complex-valued $nN\times nN$-matrix (denoted
by the same symbol) and we denote by $V=p \bbC^{nN}$ the vector space obtained by this projection acting on 
$\bbC^{nN}$. We have $V\simeq \bbC^m$ for some 
$\bbN\ni m \leq nN$ and thus
\begin{flalign*}
pA^N\simeq (p\bbC^{nN})\otimes_\bbC \bbC^n\simeq \bbC^m\otimes_\bbC\bbC^n\simeq M_{m,n}(\bbC)~.
\end{flalign*}
The right $A$-module structure on $M_{m,n}(\bbC)$ is given by matrix multiplication from the right.
\end{proof}
The endomorphism algebra $\End_A(\EE)$  is isomorphic to the complex-valued
$m\times m$-matrices, $\End_A(\EE)\simeq M_m(\bbC)$, which act on $\EE$  by matrix multiplication
from the left, $M_{m}(\bbC)\times M_{m,n}(\bbC)\to M_{m,n}(\bbC)\,,~(T,s)\mapsto T\,s$.

We equip $\EE=M_{m,n}(\bbC)$ with the standard hermitian structure
$\langle \cdot,\cdot\rangle :\EE\times\EE\to A$ defined by, for all $s,t\in \EE$,
\begin{flalign}\label{eqn:hermitgeneral}
\langle s,t\rangle := s^\ast t~,
\end{flalign}
where $\ast$ acts on $\EE$ by hermitian conjugation. 
We say that an endomorphism $T\in \End_A(\EE)$ is hermitian if, for all $s,t\in\EE$,
$\langle s,T(t)\rangle = \langle T(s),t\rangle$.
Under the isomorphism $\End_A(\EE)\simeq M_m(\bbC)$
 (anti)hermitian endomorphisms correspond to (anti)hermitian matrices.
 
A connection on a right $A$-module $\EE$ is a $\bbC$-linear map
  $\nab:\EE\to\EE\otimes_A \underline{\Omega}^1_\g(A)$ satisfying the Leibniz rule,
 for all $s\in\EE$ and $a\in A$,
 $\nab(s\,a) = (\nab s)\,a + s\otimes_A \dd a$.
Since $\underline{\Omega}^1_\g(A) \simeq A\otimes_\bbC \g^\ast$ there is the following isomorphism of right $A$-modules
$\EE\otimes_A\underline{\Omega}^1_\g(A) \simeq \EE\otimes_A (A\otimes_\bbC \g^\ast) \simeq \EE\otimes_\bbC \g^\ast$.
Thus,  we can equivalently regard a connection as a $\bbC$-linear map
$\nab: \EE \to \EE\otimes_\bbC\g^\ast$ satisfying, for all $s\in \EE$ and $a\in A$,
$\nab(s\,a) = (\nab s)\,a + s\,\dd a$. The space of all connections on $\EE$ is denoted by $\Con_A(\EE)$
and it forms an affine space over $\Hom_A(\EE,\EE\otimes_\bbC \g^\ast)\simeq \End_A(\EE)\otimes_\bbC\g^\ast
\simeq M_m(\bbC)\otimes_\bbC \g^\ast$.

 Associated to $\nab\in \Con_A(\EE)$ there is the covariant derivative along any $X\in \g$ defined
by the $\bbC$-linear map $\nab_{X}:\EE \to \EE\,,~s\mapsto \nab_{X}s = (\nab s)(X)$, where the last term denotes the canonical 
 evaluation of $\EE\otimes_\bbC \g^\ast$ on $\g$. The covariant derivative has the following properties, for all
$f,h\in\mathcal{Z}(A)$, $X,Y\in\g$, $s\in \EE$ and $a\in A$,
$\nab_{{f\,X+h\,Y}} s = f\,\nab_{X} s + h\,\nab_{Y} s$ and $\nab_X(s\,a) = (\nab_X s)\,a + s\,X(a)$.

It is well-known \cite{DuboisViolette:1988ir,DuViLecture1,DuViLecture2,Masson:2008uq,Cagnache:2008tz,Massonnew}  that there 
exists a connection on $\EE$, the so-called 
canonical connection, defined by, for all $s\in \EE$,
$\nab^\theta s := -s\,\theta$,
where $\theta\in\underline{\Omega}_\g^1(A)$ is the canonical one-form. Since $\Con_A(\EE)$ is an affine space over
 $M_m(\bbC)\otimes_\bbC \g^\ast$,
any connection $\nab\in\Con_A(\EE)$ can be written as, for all $s\in\EE$,
\begin{flalign}\label{eqn:gaugepotential}
\nab s = \nab^\theta s + B\,s = -s\,\theta  + B\,s~,
\end{flalign}
where $B\in M_m(\bbC)\otimes_\bbC \g^\ast$ is the gauge potential relative to $\nab^\theta$.

A connection is compatible with the hermitian structure (\ref{eqn:hermitgeneral}) if, for all $s,t\in\EE$ and $X\in\g$,
$X\big(\langle s,t \rangle\big) =  \langle \nab_{X^\ast} s,t \rangle + \langle s,\nab_X t\rangle$.
The canonical connection $\nab^\theta$ is compatible with the hermitian
structure (\ref{eqn:hermitgeneral})  and a general connection $\nab\in \Con_A(\EE)$ is compatible if and only if the
gauge potential $B\in M_m(\bbC)\otimes_\bbC \g^\ast$ relative to $\nab^\theta$ is antihermitian.

For all $X,Y\in\g$ we define the curvature endomorphism $F(X,Y)\in \End_A(\EE)$ by, for all $s\in\EE$,
$F(X,Y)s:=\nab_X \nab_Y s - \nab_Y\nab_X s -\nab_{[X,Y]} s$. 
Decomposing $\nab$ as in (\ref{eqn:gaugepotential}) into the canonical connection and a gauge potential 
one obtains after also using the isomorphism $\End_A(\EE)\simeq M_m(\bbC)$ an expression of
the curvature in terms of $m\times m$-matrices
\begin{flalign}\label{eqn:curvature}
F(X,Y) = \big[B(X),B(Y)\big] - B\big([X,Y]\big)~.
\end{flalign}
Notice that the connection is flat, i.e.~of vanishing curvature, if and only if the map $B:\g\to M_m(\bbC)$
 is a Lie algebra representation of $\g$.

We next introduce gauge transformations. Let us denote by $\mathcal{U}_\EE \subseteq \End_A(\EE)$ the
group of unitary endomorphisms of the module $\EE$, i.e.~for all $u\in\mathcal{U}_\EE$ and all $s,t\in\EE$,
$\langle u(s),u(t)\rangle = \langle s, t\rangle$.
Employing the isomorphism $\End_A(\EE)\simeq M_m(\bbC)$ we find that the group $\mathcal{U}_\EE$ is isomorphic
to the group of unitary $m\times m$-matrices, $\mathcal{U}_\EE \simeq U_m\subseteq M_m(\bbC)$.
We define an action of $\mathcal{U}_\EE$ on the affine space of connections
\begin{flalign}
\nonumber \mathcal{U}_\EE\times \Con_A(\EE) \to\Con_A(\EE) ~,~~(u,\nab)\mapsto \nab^u = 
\big(u\otimes_\bbC\id_{\g^\ast} \big)\circ \nab\circ u^\ast~.
\end{flalign}
For the covariant derivative this action simply reads, for all $X\in\g$,
 $(u,\nab_X) \mapsto u\circ\nab_X\circ u^\ast$, and the curvature transforms, for all $X,Y\in\g$,
  as $(u,F(X,Y))\mapsto u\, F(X,Y)\, u^\ast$, where we have implicitly used the isomorphism $\mathcal{U}_\EE \simeq U_m$.
   Notice that the canonical connection $\nab^\theta$ is gauge invariant
  and that the gauge potential $B\in M_m(\bbC)\otimes_\bbC \g^\ast$ (relative to $\nab^\theta$) 
  transforms, for all $X\in\g$, as
  \begin{flalign}\label{eqn:gaugetrafo}
  (u,B(X))\mapsto u\,B(X)\,u^\ast ~.
  \end{flalign} 
 
The following proposition is taken from \cite{DuboisViolette:1988ir,DuViLecture1,DuViLecture2,Masson:2008uq} 
and it provides the main motivation for our studies on module
parallel transport observables to be presented below.
\begin{proposition}\label{propo:gaugecopyproblem}
Let $A=M_n(\bbC)$ and let $\EE =M_{m,n}(\bbC)$ be equipped with the hermitian structure (\ref{eqn:hermitgeneral}).
  Then gauge equivalence classes of flat hermitian connections on $\EE$ are in 1-to-1 
  and onto correspondence with unitary inequivalent
Lie algebra representations $B:\g\to M_m(\bbC)$.
\end{proposition}
\begin{proof}
Follows immediately from (\ref{eqn:curvature}) and (\ref{eqn:gaugetrafo}).
\end{proof}
Considering the prime example of a free module $\EE = A$ over the fuzzy sphere \cite{Madore:1991bw}, 
where $A=M_n(\bbC)$, $n=2\,j+1$ ($j$ being the maximal spin) and $\g=\mathfrak{su}_2^\bbC$,
there are the unitary inequivalent representations $B:\mathfrak{su}_2^\bbC\to M_n(\bbC)$ 
of spin $ s= 0,\frac{1}{2},\dots,j$. 
Since all of them have the same curvature, namely zero,
we can not determine the gauge equivalence class 
of the connection from curvature observables, i.e.~gauge invariant functions
$\Con_A(\EE)\to\bbC$ constructed out of $F$. 
In other words, sets of observables constructed out of $F$ can not be separating on the
space of gauge equivalence classes of hermitian connections.

This problem 
also exists in the commutative setting, where $A$ is the algebra of smooth functions on a manifold $M$
and $\EE$ is the $A$-module of smooth sections of a hermitian vector bundle $\pi:E\to M$, see e.g.~\cite{WuYang,Driver}.
In this case the solution is to consider parallel transports along curves.
Indeed, it can be shown that, for certain gauge groups,
the Wilson loop observables constructed from such parallel transports 
are separating on the space of connections modulo gauge transformations, see e.g.~\cite{Driver,Sengupta}
and references therein.
A naive generalization of the concept of parallel transport along curves to noncommutative geometry seems to be problematic,
since firstly we in general do not have the notion of curves and secondly  there is no fibre over a ``point''.
See, however, \cite{Ishibashi:1999hs,Ambjorn:2000cs} for an interesting definition of
Wilson loop observables for noncommutative gauge theory on Moyal-Weyl deformed spaces.

In order to develop a concept of parallel transport which does not rely on curves or points, 
and hence is in principle applicable to noncommutative geometry, we propose the following
\begin{definition}\label{defi:transport}
Let $A$ be an associative and unital algebra and $\EE$ a right $A$-module.
\begin{itemize}
\item[1.)]
A {\bf one-parameter group of automorphisms} of $A$ is a map
$\varphi: \bbR \times A \to A\,,~(\tau,a) \mapsto \varphi(\tau,a) =\varphi_\tau(a)$, such that
\begin{itemize}
\item[(i)] $\varphi_\tau(a\,b) = \varphi_\tau(a) \,\varphi_\tau(b)$, for all $\tau\in\bbR$ and $a,b\in A$
\item[(ii)] $\varphi_0 =\id_A$
\item[(iii)] $\varphi_{\tau+\sigma} = \varphi_\tau\circ\varphi_\sigma$, for all $\tau,\sigma\in\bbR$
\end{itemize}
\item[2.)] Let $\varphi:\bbR \times A \to A $ be a one-parameter group of automorphisms of $A$. A {\bf module parallel transport} on
$\EE$ along $\varphi$ is a map $\Phi: \bbR \times \EE \to \EE\,,~(\tau,s)\mapsto \Phi(\tau,s)=\Phi_\tau(s)$, such that
\begin{itemize}
\item[(i)] $\Phi_\tau(s\,a) = \Phi_\tau(s)\,\varphi_\tau(a)$, for all $\tau\in\bbR$, $s\in\EE$ and $a\in A$
\item[(ii)] $\Phi_0=\id_\EE$
\item[(iii)] $\Phi_{\tau+\sigma}= \Phi_\tau\circ \Phi_\sigma$, for all $\tau,\sigma\in\bbR$
\end{itemize}
\end{itemize}
If $A$ and $\EE$ are equipped with a smooth structure, the maps $\varphi$ and $\Phi$ are required to be smooth.
\end{definition}

In the  \ref{app:commutative} we study this notion of module parallel transport in commutative differential geometry
and show its relation to the usual parallel transport along curves.
It will turn out that Definition \ref{defi:transport} 
provides a suitable generalization of parallel transport to (derivation-based) noncommutative geometry.

We now investigate the structures of Definition \ref{defi:transport} in detail 
for $A=M_n(\bbC)$ and $\EE = M_{m,n}(\bbC)$. In this case there is a bijection
between smooth one-parameter groups of automorphism of $A$ and derivations $\Der(A)$ given by the matrix exponential.
We associate to any $X\in\Der(A)$ the one-parameter group of automorphisms
 $\varphi^X:\bbR\times A\to A\,,~(\tau,a)\mapsto \varphi_\tau^X(a)=e^{\tau X}(a)$. The map $\varphi^X$ is uniquely specified
 by the ordinary differential equation
$\frac{d}{d\tau}\varphi^X_\tau = X\circ \varphi^X_\tau$,
 together with the initial condition $\varphi_0^X=\id_A$, which is part of Definition \ref{defi:transport}.
The derivation $X$ can be determined from $\varphi^X$ by taking the $\tau$-derivative at $\tau=0$.
Furthermore, using the canonical one-form $\theta\in\underline{\Omega}_\g^1(A)$, we can express, for all
$X\in \g\subseteq \Der(A)$ and $a\in A$, the derivation
as the difference between left and right multiplication by $\theta(X)$,
$X(a) = \big[\theta(X),a\big] = \theta(X)\,a - a\,\theta(X)  $.
Since left and right multiplication commutes, the exponential becomes, for all $X\in\g$, $a\in A$ and $\tau\in\bbR$,
\begin{flalign}\label{eqn:1pg}
\varphi_\tau^X(a) = e^{\tau\theta(X)} \,a\,e^{-\tau\theta(X)}~.
\end{flalign}
Notice that $\varphi^X$ is a one-parameter group of $\ast$-automorphisms if and only if $X$ is a real derivation.
\begin{proposition}\label{propo:parallelexistence}
Let $A=M_n(\bbC)$, $\EE=M_{m,n}(\bbC)$ and $\nab\in\Con_A(\EE)$. Then there exists  for all $X\in\g\subseteq \Der(A)$
a unique module 
parallel transport $\Phi^{\nab_X}:\bbR\times\EE\to\EE$ along $\varphi^X$ satisfying the ordinary differential equation
\begin{flalign}\label{eqn:difeqn}
\frac{d}{d\tau}\Phi^{\nab_X}_\tau = \nab_X\circ \Phi_\tau^{\nab_X}~.
\end{flalign} 
It is given by, for all $X\in\g$, $s\in\EE$ and $\tau\in\bbR$,
\begin{flalign}\label{eqn:hol}
\Phi_\tau^{\nab_X}(s) = e^{\tau B(X)}\,s\,e^{-\tau \theta(X)}~,
\end{flalign}
where $B\in M_m(\bbC)\otimes_\bbC \g^\ast$ is the gauge potential relative to the canonical connection $\nab^\theta$.
\end{proposition}
\begin{proof}
The differential equation (\ref{eqn:difeqn}) subject to the initial condition
 $\Phi^{\nab_X}_0=\id_\EE$, which is part of Definition
 \ref{defi:transport}, has a unique solution given by
$\Phi^{\nab_X}_\tau = e^{\tau\nab_X}$,
where we regard $\nab_X:\EE\to\EE$ as an element of $\End_\bbC(\EE)\simeq M_{n m}(\bbC)$.
Using the decomposition  (\ref{eqn:gaugepotential}) we have, for all $X\in \g$ and $s\in\EE$,
$\nab_X(s) = -s\,\theta(X) + B(X)\,s $.
Thus, for all $X\in \g$, $s\in\EE$ and $\tau\in\bbR$, we obtain (\ref{eqn:hol}).
From (\ref{eqn:hol}) and (\ref{eqn:1pg}) one easily checks that $\Phi^{\nab_X}$ is a module parallel transport along $\varphi^X$, i.e.~that
it satisfies the conditions posed in Definition  \ref{defi:transport}.
\end{proof}
Notice that from $\Phi^{\nab_X}$  the covariant derivative $\nab_X$ can be determined  by taking the $\tau$-derivative
at $\tau=0$. Thus, the complete connection $\nabla$ can be determined from the set  $\big\{\Phi^{\nab_X}\big\}_{X\in \g}$.

We come to the definition of a class of gauge invariant observables, i.e.~gauge invariant functions
$\Con_A(\EE)\to\bbC$, which is separating on the space of hermitian connections modulo gauge transformations.
We can construct for {\it any} $\tau\in\bbR$ and $X\in\g$ an 
endomorphism $\widetilde{\Phi}^{\nab_X}_\tau\in\End_A(\EE)$  by composing $\Phi^{\nab_X}_\tau$ with the
right multiplication of $e^{\tau\theta(X)}\in A$, for all $\tau\in \bbR$, $X\in\g$ and $s\in\EE$,
\begin{flalign}
\nonumber \widetilde{\Phi}_\tau^{\nab_X}(s) := \Phi_\tau^{\nab_X}(s)\,e^{\tau\theta(X)} 
= e^{\tau B(X)}\,s ~,
\end{flalign}
where in the second equality we have used (\ref{eqn:hol}).
Applying the trace
$\Tr :\End_A(\EE)\to\bbC\,,~T\mapsto \Tr(T) = \Tr_{M_m(\bbC)}\big(T\big)$ on
products of the endomorphisms $\widetilde{\Phi}_\tau^{\nab_X}$ (without loss of generality we set $\tau=1$), 
we obtain the gauge invariant observables
\begin{flalign}\label{eqn:observablescorrect}
\mathcal{W}_{(X_1,\dots, X_N)} :\Con_A(\EE)  \to\bbC\,,\,
\nab \mapsto \mathcal{W}_{(X_1,\dots,X_N)}(\nab) 
= \Tr\left(\widetilde{\Phi}_1^{\nab_{X_1}}\cdots \widetilde{\Phi}_1^{\nab_{X_N}} \right)
\end{flalign}
labeled by $(X_1,\dots,X_N)\in\g^{N}$ and $N\in \mathbb{N}$. 
Note that for defining the observables (\ref{eqn:observablescorrect}) we {\it did not} require the automorphisms
 $\varphi_\tau^{X_{i}}$, $i=1,\dots,N$, to be periodic in $\tau$.
This is similar to the gauge invariant {\it open} Wilson lines discovered in noncommutative
gauge theory on the Moyal-Weyl space, see e.g.~\cite{Ishibashi:1999hs,Ambjorn:2000cs}.

Let us define the set of basic observables by
\begin{flalign}\label{eqn:obssetgeneral}
\mathcal{O} := \Big\{\mathcal{W}_{(X_1,\dots,X_N)} : (X_1,\dots,X_N)\in\g_\bbR^{N}~,~~N\in\mathbb{N}\Big\}~.
\end{flalign}
As a main result we can solve the gauge copy problem for fuzzy gauge theories due to the following
\begin{theorem}\label{theo:general}
Let $A=M_n(\bbC)$ and let $\EE=M_{m,n}(\bbC)$ be a finite projective right $A$-module,
equipped with the hermitian structure (\ref{eqn:hermitgeneral}).
Then the set of observables $\mathcal{O}$ (\ref{eqn:obssetgeneral}) is separating on the space of hermitian  connections
modulo gauge transformations.
\end{theorem}
\begin{proof}
Let $\nab,\nab^\prime \in\Con(\EE)$ be two hermitian connections.
Using a basis $\{e_i\in \g_\bbR :i=1,\dots,\dim(g_\bbR)\}$ of $\g_\bbR$, the 
connection $\nab^{(\prime)}$ is uniquely specified by the components $B_i^{(\prime)} := B^{(\prime)}(e_i)$
of the gauge potential $B^{(\prime)}$ relative to $\nab^\theta$, which are antihermitian $m\times m$-matrices.
We have to show that  $\mathcal{W}_{(X_1,\dots,X_N)}(\nab) = \mathcal{W}_{(X_1,\dots,X_N)}(\nab^\prime)$, for all
$(X_1,\dots,X_N)\in\g_\bbR^{N}$ and $N\in\mathbb{N}$, implies that $\nab$ and $\nab^\prime$ are
gauge equivalent. On the level of gauge potentials this means that there exists
$u\in U_m$, such that $B_i^\prime = u\,B_i\,u^\ast$, for all $i$, cf.~(\ref{eqn:gaugetrafo}).
This statement is shown as follows: Taking derivatives along $\tau_1,\dots,\tau_N$ of the observables
$\mathcal{W}_{(\tau_1 \,X_1,\dots,\tau_N\,X_N)}$ at $\tau_1=\dots=\tau_N=0$ we recover the traces
of any monomial in the antihermitian matrices $B_i^{(\prime)}$, which according to the condition
 $\mathcal{W}_{(X_1,\dots,X_N)}(\nab) = \mathcal{W}_{(X_1,\dots,X_N)}(\nab^\prime)$, for all 
 $(X_1,\dots,X_N)\in\g_\bbR^{N}$ and $N\in\mathbb{N}$, have to agree. 
 The fundamental theorems of invariant theory, see e.g.~\cite[Statement (5)]{invariant1} and 
 \cite{invariant2} for details, then imply that $B_i^\prime = u\,B_i\,u^\ast$, for all $i$,
 and hence that $\nab$ and $\nab^\prime$ are gauge equivalent.
\end{proof}

\section*{Acknowledgments}
I would like to thank Hanno Gottschalk, Thierry Masson, Thorsten Ohl, 
Christian S\"amann, Peter Schupp, Richard Szabo 
and Christoph Uhlemann for useful discussions and comments.
Furthermore, I am very grateful to the referee for constructive comments
and to Markus Reineke for providing me with help to implement them.


\appendix

\section{\label{app:commutative} Module parallel transport in commutative differential geometry}
Let $M$ be a smooth manifold, which we assume for the moment to be compact.
Let further $\pi: E\to M$ be a smooth complex vector bundle over $M$ and $\nab$ a connection
on $E$. Given any smooth path $\gamma: [\tau_1,\tau_2]\to M$, with $\tau_1<\tau_2\in\bbR$, 
we can use the connection $\nab$ to construct
a parallel transport map $h^\nab_\gamma:E_{\gamma(\tau_1)}\to E_{\gamma(\tau_2)}$ from the fibre 
over the starting point $\gamma(\tau_1)\in M$ to the fibre over the ending point $\gamma(\tau_2)\in M$ of the path $\gamma$.
 The map $h^\nab_\gamma$ is a vector space isomorphism and it describes how a vector in $E_{\gamma(\tau_1)}$ 
 is parallel transported along the curve $\gamma$, yielding a vector in $E_{\gamma(\tau_2)}$. 

The goal is to induce from this local concept of parallel transport along curves a global one.
In view of the structures available in noncommutative geometry, we intend to derive a notion
of parallel transport of sections  of $E$ along flows on the algebra $C^\infty(M)$.
For this let us first consider
a one-parameter group of diffeomorphisms $f:\bbR\times M\to M$, i.e.~$f_\tau:=f(\tau,\,\cdot\,)$ is a diffeomorphism for all
$\tau\in\bbR$, $f_0 = \id_M$ and $f_\tau\circ f_\sigma = f_{\tau+\sigma}$, for all $\tau,\sigma\in\bbR$.
Notice that fixing some point $p\in M$ and some interval $[\tau_1,\tau_2]\subset \bbR$ we obtain
a path by setting $[\tau_1,\tau_2]\to M \,,~\tau \mapsto  f_\tau(p)=f(\tau,p)$, with starting point $f(\tau_1,p)$
 and ending point $f(\tau_2,p)$.
Using the connection $\nab$ we can  lift any one-parameter group of diffeomorphisms $f$ to a one-parameter
group of vector bundle automorphisms $f^\nab:\bbR\times E \to E$, such that the following diagram commutes:
\begin{flalign}\label{eqn:comdiagapp}
\xymatrix{
\ar[d]_-{\id_\bbR\times \pi}\bbR\times E \ar[rr]^-{f^\nab}& & E \ar[d]^-{\pi}\\
\bbR\times M \ar[rr]^-{f} && M
} 
\end{flalign}
Explicitly, the map $f^\nab$ reads, for all $\tau\in\bbR$ and $e\in E$,
$f^\nab(\tau,e):= h^\nab_{\gamma_{\tau,e}}(e)$, with the path given by 
$\gamma_{\tau,e}:[0,\tau]\to M\,,~\sigma \mapsto f(\sigma,\pi(e))$.

We denote the algebra of smooth complex-valued functions on $M$ by $A:= C^\infty(M)$
and the $A$-module of smooth complex-valued sections of $E$ by $\EE:= \Gamma^\infty(E)$.
We can canonically lift $f$ and $f^\nab$ to maps on $A$ and $\EE$.
On $A$ we obtain the one-parameter group of algebra automorphisms given by 
\begin{flalign*}
\varphi : \bbR\times A\to A\,,~(\tau,a) \mapsto \varphi_\tau(a) = a\circ f_{-\tau}~,
\end{flalign*}
and on $\EE$ the one-parameter group of $\bbC$-linear maps
\begin{flalign*}
\Phi^\nab : \bbR \times \EE \to \EE \,,~(\tau,s)\mapsto \Phi^\nab_\tau(s) = f^\nab_\tau\circ s\circ f_{-\tau}~.
\end{flalign*}
Using the commutative diagram (\ref{eqn:comdiagapp}) we can easily verify the following 
compatibility condition between $\varphi$, $\Phi^\nab$ and the right
$A$-module structure on $\EE$, for all $\tau\in\bbR$, $s\in \EE$ and $a\in A$,
\begin{flalign*}
\Phi_\tau^\nab(s\cdot a) = \Phi_\tau^\nab(s) \cdot \varphi_\tau(a)~.
\end{flalign*}
The reader should compare this equation with the second part of Definition \ref{defi:transport}.
We can associate to the one-parameter group of diffeomorphisms $f$ the generating vector field
$X\in \Der(A)$ by taking the $\tau$-derivative of $\varphi_\tau$ at $\tau=0$. Taking the $\tau$-derivative
of $\Phi^\nab_\tau$ at $\tau=0$ yields the covariant derivative $\nab_X$ along $X$. More generally, taking the
$\tau$-derivative of $\varphi_\tau$ and $\Phi^\nab_\tau$ leads to the differential equations
\begin{flalign}\label{app:dgl}
\frac{d}{d\tau} \varphi_\tau = X\circ \varphi_\tau\quad,\qquad \frac{d}{d\tau}\Phi^\nab_\tau = \nab_X\circ \Phi^\nab_\tau~.
\end{flalign}

The above discussion shows that, starting from a one-parameter group of diffeomorphisms of $M$, 
we can construct the pair of maps $(\varphi,\Phi^\nab)$. Furthermore, starting from a one-parameter group of algebra
automorphisms $\varphi:\bbR\times A\to A$, we can integrate the generating vector field 
$X(\,\cdot\,) = \frac{d}{d\tau} \varphi_\tau(\,\cdot\,) \vert^{}_{\tau=0}$ to yield a one-parameter group
of diffeomorphisms of $M$, which then can be used to construct $\Phi^\nab$.
In other words, given a connection $\nab$ on $E$ we can construct for every one-parameter group of
algebra automorphisms $\varphi$ a module parallel transport $\Phi^\nab$ according to Definition \ref{defi:transport}.
Because of (\ref{app:dgl}) the covariant derivative $\nab_X$ along the generator $X$ of $\varphi$ can be reconstructed from this
transport. 
Since we have assumed $M$ to be compact, every vector field $X$ can be integrated to yield a one-parameter
group of algebra automorphisms, and thus the covariant derivative $\nab_X$ along every vector field $X$
can be reconstructed from module parallel transports. This specifies uniquely the connection $\nab$ on $E$.

Let us briefly discuss the situation where $M$ is noncompact. In this case, there can be non-complete vector fields
which do not give rise to global one-parameter groups of diffeomorphisms of $M$. However, every compactly
supported vector field is complete, and can be integrated to a global one-parameter group of diffeomorphisms,
which gives rise to the pair of maps $(\varphi,\Phi^\nab)$. Thus, we can reconstruct the covariant derivative $\nab_X$
along every compactly supported vector field $X$ from module parallel transports, which is sufficient for uniquely specifying the
connection.



\begin{thebibliography}{10}
\bibitem{Landi:1997sh} 
  G.~Landi,
  ``An Introduction to noncommutative spaces and their geometry,''
  Lecture Notes in Physics: Monographs, m51 (Springer-Verlag, Berlin Heidelberg, 1997) ISBN 3-540-63509-2
  [hep-th/9701078].
  
\bibitem{Aschieri:2005zs} 
  P.~Aschieri, M.~Dimitrijevic, F.~Meyer and J.~Wess,
  ``Noncommutative geometry and gravity,''
  Class.\ Quant.\ Grav.\  {\bf 23}, 1883 (2006)
  [hep-th/0510059].
  
  
  
\bibitem{Woronowicz:1989rt} 
  S.~L.~Woronowicz,
  ``Differential calculus on compact matrix pseudogroups (quantum groups),''
  Commun.\ Math.\ Phys.\  {\bf 122}, 125 (1989).
  
\bibitem{DuVi:88}
M.~Dubois-Violette,
``D{\'e}rivations et calcul differentiel non commutatif,''
 C.R. Acad. Sci. Paris, S{\'e}rie {I}, 307:403--408 (1988).

 \bibitem{DuViMich:94}
M.~Dubois-Violette and P.~W.~Michor,
``D{\'e}rivations et calcul differentiel non commutatif {II},''
C.R. Acad. Sci. Paris, S{\'e}rie {I}, 319:927--931 (1994).

\bibitem{DuViMich:96}
M.~Dubois-Violette and P.~W.~Michor,
``Connections on central bimodules in noncommutative differential
  geometry,''
 Journal of Geometry and Physics, 20:218--232 (1996) [arXiv:q-alg/9503020].


\bibitem{DuViLecture1}
M.~Dubois-Violette,
``Non-commutative differential geometry, quantum mechanics and gauge theory,''
in Differential Geometric Methods in Theoretical Physics, C.~Bartocci et al (eds), 1991 Springer Verlag.


\bibitem{DuViLecture2}
M.~Dubois-Violette,
``Lectures on graded differential algebras and noncommutative geometry,''
in Noncommutative Differential Geometry and Its Application to Physics,
Proceedings of the Workshop Shonan, Japan, June 1999, Y.~Maeda, H.~Moriyoshi et al (eds),
Kluwer Academic Publishers 2001, pp.~245-306 [arXiv:math/9912017].


\bibitem{Masson:2008uq} 
  T.~Masson,
  ``Examples of derivation-based differential calculi related to noncommutative gauge theories,''
  Int.\ J.\ Geom.\ Meth.\ Mod.\ Phys.\  {\bf 5}, 1315 (2008)
  [arXiv:0810.4815 [math-ph]].
  
\bibitem{DuboisViolette:1988ir} 
  M.~Dubois-Violette, R.~Kerner and J.~Madore,
  ``Noncommutative Differential Geometry of Matrix Algebras,''
  J.\ Math.\ Phys.\  {\bf 31}, 316 (1990).
  
\bibitem{Madore:1991bw} 
  J.~Madore,
  ``The Fuzzy sphere,''
  Class.\ Quant.\ Grav.\  {\bf 9}, 69 (1992).
  
  
\bibitem{Landi:2006nb} 
  G.~Landi and W.~van Suijlekom,
  ``Noncommutative instantons from twisted conformal symmetries,''
  Commun.\ Math.\ Phys.\  {\bf 271}, 591 (2007)
  [math/0601554 [math-qa]].
  
\bibitem{Grosse:2001qt} 
  H.~Grosse, C.~W.~Rupp and A.~Strohmaier,
  ``Fuzzy line bundles, the Chern character and topological charges over the fuzzy sphere,''
  J.\ Geom.\ Phys.\  {\bf 42}, 54 (2002)
  [math-ph/0105033].

  

  
\bibitem{WuYang}
  T.~T.~Wu and  C.~N.~Yang, 
  ``Some remarks about unquantized non-abelian gauge fields,''
  Phys.\ Rev.\ D {\bf 12}, 3843-3844 (1975). 

 \bibitem{Driver}
 B.~Driver,
  ``Classifications of bundle connection pairs by parallel translation and lassos,''
J.\ Funct.\ Anal.\ {\bf 83}, no.\ 1, 185–231 (1989).   
  
\bibitem{Sengupta}
A.~Sengupta,  
  ``Gauge invariant functions of connections,''
  Proc.\ Am.\ Math.\ Soc.\ {\bf 121}, 897-905 (1994).
  
\bibitem{Ishibashi:1999hs} 
  N.~Ishibashi, S.~Iso, H.~Kawai and Y.~Kitazawa,
  ``Wilson loops in noncommutative Yang-Mills,''
  Nucl.\ Phys.\ B {\bf 573}, 573 (2000)
  [hep-th/9910004].
  
\bibitem{Ambjorn:2000cs} 
  J.~Ambjorn, Y.~M.~Makeenko, J.~Nishimura and R.~J.~Szabo,
  ``Lattice gauge fields and discrete noncommutative Yang-Mills theory,''
  JHEP {\bf 0005}, 023 (2000)
  [hep-th/0004147].
  
 \bibitem{Koszul}
 J.~L.~Koszul,
 ``Lectures on fibre bundles and differential geometry,''
 Notes by S.~Ramanan. Tata Institute of Fundamental Research
  Lectures on Mathematics,~No.~20, Tata Institute of Fundamental Research (1965).

\bibitem{Cagnache:2008tz} 
  E.~Cagnache, T.~Masson and J.~-C.~Wallet,
  ``Noncommutative Yang-Mills-Higgs actions from derivation-based differential calculus,''
  J.\ Noncommut.\ Geom.\  {\bf 5}, 39 (2011)
  [arXiv:0804.3061 [hep-th]].
  
\bibitem{Massonnew}
T.~Masson,
``Gauge theories in noncommutative geometry,''
AIP Conf.\ Proc.\  {\bf 1446}, 73 (2010)
[arXiv:1201.3345 [math-ph]].  

\bibitem{invariant1}
C.~Procesi,
``The invariants of nxn matrices,''
Bull.\ Amer.\ Math.\ Soc.\ {\bf 82}, no.\ 6, 891-892 (1976).

\bibitem{invariant2}
C.~Procesi,
``The invariant theory of nxn matrices,''
Advances in Math.\ {\bf 19}, no.\ 3, 306-381 (1976). 

  
\end{thebibliography}
\end{document}